\documentclass[12pt,a4paper]{article}
\usepackage[text={16.5cm,23cm},centering]{geometry}
\usepackage{amsmath}
\usepackage{amsthm}
\usepackage[utf8]{inputenc}
\usepackage{tgtermes,tgheros,tgcursor}
\usepackage[varg]{newtxmath}
\usepackage{hyperref}
\hypersetup{%
 colorlinks=true,%
 linkcolor=blue,
 citecolor=blue,
}

\numberwithin{equation}{section}

\newcommand{\pdv}[1]{\frac{\partial}{\partial #1}}
\newcommand{\del}{\partial}
\newcommand{\Ncal}{\mathcal{N}}
\newcommand{\Zb}{\mathbb{Z}}

\newcommand{\Qd}{\mathcal{Q}}
\newcommand{\Dd}{\mathcal{D}}
\newcommand{\Lcal}{\mathcal{L}}
\newcommand{\Ls}{\mathscr{L}}
\newcommand{\Jq}{\Psi}

\newcommand{\susyder}{\mathbb{Q}}
\newcommand{\Pxyz}{P_{xyz}}
\newcommand{\bt}{\tilde{b}}
\newcommand{\Bt}{\widetilde{B}}
\newcommand{\phit}{\tilde{\phi}}
\newcommand{\Phit}{\widetilde{\Phi}}
\newcommand{\mut}{\tilde{\mu}}
\newcommand{\wt}{\tilde{w}}

\newcommand{\kv}{\vec{k}}
\newcommand{\xv}{\vec{x}}
\newcommand{\zv}{\vec{0}}
\newcommand{\nv}{\vec{n}}

\newtheorem{theorem}{Theorem}

\begin{document}
\begin{center}
  \begin{flushright}
    OU-HET 1090
  \end{flushright}
  \vspace{8ex}
  {\Large \bfseries \boldmath Supersymmetric quantum field theory with exotic symmetry in $3+1$ dimensions and fermionic fracton phases}\\
  \vspace{4ex}
  {\Large Satoshi Yamaguchi}\\
  \vspace{2ex}
  {\itshape Department of Physics, Graduate School of Science, 
  \\
  Osaka University, Toyonaka, Osaka 560-0043, Japan}\\
  \vspace{1ex}
  \texttt{yamaguch@het.phys.sci.osaka-u.ac.jp}\\
  \begin{abstract}
    We propose a supersymmetric quantum field theory with exotic symmetry related to fracton phases.  We use superfield formalism and write down the action of a supersymmetric version of the $\varphi$ theory in $3+1$ dimensions.  It contains a large number of ground states due to the fermionic higher pole subsystem symmetry.  Its residual entropy is proportional to the area instead of the volume.  This theory has a self-duality similar to that of the $\varphi$ theory.  We also write down the action of a supersymmetric version of a tensor gauge theory, and discuss BPS fractons.
  \end{abstract}
\end{center}

\vspace{4ex}
\section{Introduction and summary}
Recently, fracton phases are attracting a lot of attention. For review, refer to \cite{Nandkishore:2018sel,Pretko:2020cko} and references therein.
Such a system shows exotic properties, such as sub-extensive entropy, local particle-like excitation with restricted mobility, and so on.
There are two approaches to fracton phases: solvable lattice models and continuum description.  In this paper, we investigate continuum description.

One of the most important developments in this topic is the discovery of higher pole subsystem symmetries.  
It has been pointed out by \cite{Pretko:2016kxt,Pretko:2016lgv} that higher pole symmetries play a key role in the exotic properties of fracton phases. 
Higher pole symmetries have been further investigated by a lot of papers including \cite{2018PhRvB..98c5111M,Bulmash:2018lid,You:2019cvs,You:2019bvu,Seiberg:2019vrp,Seiberg:2020bhn,Seiberg:2020wsg,Seiberg:2020cxy,Gorantla:2020xap}.
Higher pole symmetries are also important in this paper.

One open question raised in \cite{Pretko:2020cko} is whether there are intrinsically fermionic fracton phases.  We want to approach this question from the perspective of continuum description.  However, we do not have nice guiding principles, since fracton systems do not have the Lorentz symmetry nor continuous spacial rotational symmetry.  Therefore, in this paper, we propose supersymmetry that relates bosons and fermions. We employ this supersymmetry as a guiding principle to introduce fermions to the system.  We expect that the fermionic part of this supersymmetric theory is a natural fermionic fracton system.

In this paper, we first propose a supersymmetrization of the $\varphi$ theory\cite{You:2019cvs,Gorantla:2020xap} in $3+1$ dimensions.  
Let us explain this $\varphi$ theory in our notation.  Let $t$ and $\xv=(x,y,z)$ be time and space coordinates, respectively.  
We impose a periodic boundary conditions in the $x,y,z$ directions for all fields throughout this paper.  
It is convenient to introduce the differential operators
\begin{align}
  \del_{\pm}=\frac12(\del_{t}\pm \alpha \del_{x}\del_{y}\del_{z}),\label{delpm}
\end{align}
where $\alpha$ is a real positive constant.
The $\varphi$ theory includes a single scalar field $\phi$ with the periodicity $\phi\sim \phi+2\pi$. The Lagrangian density of the $\varphi$ theory is given by
\begin{align}
  \Lcal_{\varphi}=2\mu_0\del_{-}\phi \del_{+}\phi
  =\frac{\mu_0}{2}[(\del_{t}\phi)^2-\alpha^2 (\del_{x}\del_{y}\del_{z} \phi)^2],\label{phitheory}
\end{align}
where $\mu_0$ is a real positive constant.
This theory has momentum and winding quadrupole symmetry and exhibits a self-duality as shown in \cite{Gorantla:2020xap}.  

The strategy for supersymmetrization of the $\varphi$ theory in this paper is as follows.  
The Lagrangian density \eqref{phitheory} resembles to that of $(1+1)$-dimensional massless free scalar theory in which $\del_{\pm}$ are replaced by the derivatives in light-cone directions.
Therefore we try to supersymmetrize the $\varphi$ theory in the same way as $1+1$ dimensions as if $\del_{\pm}$ were the derivatives in the light-cone directions.  
We employ the superfield formalism.  
In this paper, we only consider an analogue of $\Ncal=(1,1)$ supersymmetry in $1+1$ dimensions.

Here is the short summary of the results of this paper.  Our strategy works for some cases including the supersymmetric version of the $\varphi$ theory.  We find that the supersymmetric $\varphi$ theory is obtained by adding ``the $\psi$ theory''
\begin{align}
  &\Lcal_{\psi}
  =\mu_0(i\psi_{+}\del_{-}\psi_{+}+i\psi_{-}\del_{+}\psi_{-})
  =\frac{\mu_0}{2}[i\psi_{+}(\del_{t}- \alpha \del_{x}\del_{y}\del_{z})\psi_{+}+i\psi_{-}(\del_{t}+\alpha \del_{x}\del_{y}\del_{z})\psi_{-}],\nonumber\\
  &\qquad
  (\psi_{\pm}:\text{ real fermionic fields.})\label{psitheory}
\end{align}
to the $\varphi$ theory \eqref{phitheory}.  Therefore we claim that this $\psi$ theory is a natural intrinsically fermionic fracton system.

In the supersymmetric $\varphi$ theory, $\psi_{\pm}$ are the superpartners of the quadrupole symmetry currents.  Therefore they satisfy the conservation laws
\begin{align}
  (\del_{t}\mp\alpha \del_{x}\del_{y}\del_{z})\psi_{\pm}=0,
\end{align}
which are nothing but the equations of motion derived from \eqref{psitheory}.  We find that fermionic quadrupole charges
\begin{align}
  \oint dx^{i}\psi_{\pm},\quad (x^{i}=x,y,z)
\end{align}
are conserved by the same argument as usual higher pole symmetries.  Due to these fermionic charges, the residual entropy or log of the ground state degeneracy is proportional to the area of the system instead of the volume.

We also show that our supersymmetric $\varphi$ theory exhibits the self-duality, as the $\varphi$ theory does\cite{Gorantla:2020xap}.  We use our superfield formalism to show this self-duality.

We also formulate the supersymmetric tensor gauge theory that appears in gauging the global part of quadrupole symmetry of the supersymmetric $\varphi$ theory.  This is the supersymmetric version of the tensor gauge theory considered in \cite{You:2019cvs,You:2019bvu,Gorantla:2020xap}.  The multiplet of the gauge field includes two real fermions $\lambda_{\pm}$ and a real scalar $\sigma$ in addition to the tensor gauge field $B_{0},\ B_{xyz}$.  

There are many future issues that may lead to some intriguing results.  One issue is supersymmetry as a subsystem symmetry.  Unfortunately, our supersymmetry in this paper is only global supersymmetry.  It would be very interesting if one finds subsystem supersymmetry by improving our result.

Another issue is a lattice fermionic system described by the $\psi$ theory \eqref{psitheory} in the low energy limit.  
Lattice supersymmetry is a very difficult problem, but it may be feasible to find a lattice realization of the fermionic part.  We make an attempt in appendix \ref{app:lattice}.  
There have been a few studies of fermionic fracton phases from lattice models\cite{you2019building,Tantivasadakarn:2020lhq,Shirley:2020ass}.

It is also important to investigate the physical properties of the fermionic system we obtain in this paper. The conductivity, the specific heat, and so on may have some special behavior due to the fermionic nature of this system.  Our fermionic system is gapless, and therefore it will appear some critical system.  It may be useful to discuss critical exponents in our fermionic system.

There are several other issues. One is constructing interacting supersymmetric theory, which is not possible in the formulation in this paper.  
It will be also interesting to construct theories with extended supersymmetry, for example, $\Ncal=(2,2)$ and $\Ncal=(2,0)$. 
Finally, the analysis of fractons in our supersymmetric system is also a big issue.

The construction of this paper is as follows.  In section \ref{sec:susyphitheory}, we use superfields to formulate the supersymmetric $\varphi$ theory.  
We discuss how to write down the supersymmetric action.
We also show that this supersymmetric $\varphi$ theory has a self-duality.  In section \ref{sec:gaugetheory}, we consider supersymmetric tensor gauge theory.
We write down the action by using the superfield formalism.  We discuss BPS defects as fractons.
In appendix \ref{app:lattice}, we give an fermionic lattice model.
We count the ground state degeneracy and show the residual entropy is proportional to the area instead of the volume.

\section{Superfields and supersymmetric $\varphi$ theory}
\label{sec:susyphitheory}

In this section, we introduce superspace and superfields in order to write down the supersymmetric action.  In particular, we write down the supersymmetric $\varphi$ theory.  This theory has fermionic quadrupole symmetry in addition to bosonic quadrupole symmetry of the $\varphi$ theory.  We show this supersymmetric $\varphi$ theory has a self-duality.

\subsection{Superspace and superfields}

The formulation is almost parallel to $(1+1)$-dimensional $\Ncal=(1,1)$ supersymmetry.  
For supersymmetry in $(1+1)$-dimensions, refer to \cite{Hori:2003ic}, for exmaple.  
The reference \cite{Gates:1983nr} on $(2+1)$-dimensional supersymmetry is also useful since $(1+1)$-dimensional $\Ncal=(1,1)$ superspace is obtained by the dimensional reduction from $(2+1)$-dimensional $\Ncal=1$ superspace.

We introduce real fermionic coordinates $\theta^{+},\theta^{-}$ in addition to the spacetime coordinates $t,\xv=(x,y,z)$. Then we define the following derivatives in order to describe the supersymmetry transformation.
\begin{align}
  \Qd_{\pm} =-i \pdv{\theta^{\pm}} + 2 \theta^{\pm} \del_{\pm},
  \qquad
  \Dd_{\pm} =-i \pdv{\theta^{\pm}} - 2 \theta^{\pm} \del_{\pm},
\end{align}
where $\del_{\pm}=\frac12(\del_{t}\pm \alpha\del_{x}\del_{y}\del_{z})$ are differential operators of \eqref{delpm}. 
Then the anti-commutation relation between them are given by
\begin{equation}
  \begin{aligned}
    &\{\Qd_{\pm},\Qd_{\pm}\}=-4i \del_{\pm},\qquad \{\Qd_{+},\Qd_{-}\}=0,\\
    &\{\Dd_{\pm},\Dd_{\pm}\}=4i \del_{\pm},\qquad \{\Dd_{+},\Dd_{-}\}=0,\qquad \{\Qd_{\alpha},\Dd_{\beta}\}=0,\quad (\alpha,\beta = \pm).
  \end{aligned}    
\end{equation}
The last anti-commutation relation is important when we write down the action.

A real superfield is written as
\begin{align}
  \Phi(t,\xv,\theta^{+},\theta^{-})
  =\phi(t,\xv)
  +i\theta^{+}\psi_{+}(t,\xv)
  +i\theta^{-}\psi_{-}(t,\xv)
  +i\theta^{+}\theta^{-}f(t,\xv).
\end{align}
$\phi,\psi_{\pm},f$ are fields in the spacetime called ``components.'' If $\Phi$ is bosonic, $\phi$ and $f$ are real bosonic fields and $\psi_{\pm}$ are real fermionic fields.  On the other hand, if $\Phi$ is fermionic, $\phi$ and $if$ are real fermionic fields and $i\psi_{\pm}$ are real bosonic fields.

We define the supersymmetry transformation by
\begin{align}
  \delta \Phi = \susyder \Phi := (i\epsilon_{-}\Qd_{+}-i\epsilon_{+}\Qd_{-})\Phi,
  \label{susytransf}
\end{align}
where $\epsilon_{\pm}$ are infinitesimal fermionic parameters of the transformation.  In terms of components, the supersymmetry transformation is written as
\begin{equation}
\begin{aligned}
  &\delta \phi =i\epsilon_{-}\psi_{+}-i\epsilon_{+}\psi_{-},\\
  &\delta \psi_{+}=-2\epsilon_{-}\del_{+}\phi-\epsilon_{+} f,\\
  &\delta \psi_{-}=2\epsilon_{+}\del_{-}\phi-\epsilon_{-} f,\\
  &\delta f = 2i\epsilon_{-}\del_{+}\psi_{-}+2i\epsilon_{+}\del_{-}\psi_{+}.
\end{aligned}
\end{equation}
Let us call $\Phi$ a ``superfield'' if it is transformed as \eqref{susytransf} by the supersymmetry transformation.  
If $\Phi$ is a superfield, the derivatives $\Dd_{\pm} \Phi$ are also superfields, since $\Qd_{\pm}$ and $\Dd_{\pm}$ anti-commute to each other, in the same way as usual supersymmetry.  
On the other hand, if $\Phi_1$ and $\Phi_2$ are superfields, the product $\Phi_3:=\Phi_1\Phi_2$ \emph{is not a superfield.}  
In other words,  $\Phi_3$ does not follow the transformation law \eqref{susytransf} due to the third-order derivatives in $\Qd_{\pm}$.  
This property is quite different from usual supersymmetry and the main obstacle in our formulation.

Before constructing the action, let us look at the supersymmetry algebra.  Suppose we have a theory invariant under the transformation \eqref{susytransf}.  This theory includes supersymmetry generators $Q_{\pm}$ that satisfy the relation
\begin{align}
  \delta \Phi = [\epsilon_{-}Q_{+}-\epsilon_{+}Q_{-},\Phi].
\end{align}
Then, we find anti-commutation relations of $Q_{\pm}$ given by\footnote{The c-number ambiguity in the first equation is absorbed into the c-number shift of $H$ and $P_{xyz}$.}
\begin{align}
  \{Q_{\pm},Q_{\pm}\}=2H\pm 2\alpha \Pxyz, \qquad \{Q_{+},Q_{-}\}=Z, \qquad [Z,(\text{all operators})]=0. \label{SUSYalgebra}
\end{align}
Here $H$ is the Hamiltonian, and $\Pxyz$ is the generator of the transformation
\begin{align}
  \delta_{xyz}\Phi :=\epsilon \del_{x}\del_{y}\del_{z}\Phi = [i\epsilon \Pxyz,\Phi],\label{Pxyz}
\end{align}
where $\epsilon$ is a infinitesimal bosonic parameter.
Notice that $\Pxyz$ is different from the product of momenta $P_{x}P_{y}P_{z}$.  We conclude that this supersymmetry does not exist unless the theory is invariant under the transformation \eqref{Pxyz}.

\subsection{The action of the supersymmetric \texorpdfstring{$\varphi$}{phi} theory}

The difficulty in our formulation is that, unlike the ordinary supersymmetry, the product of superfields does not become a superfield.  In spite of this difficulty, the following theorem allows us to write the action of a free field theory.
\begin{theorem}\label{thm}
  Let $\Phi_1, \Phi_2$ be superfields which are transformed by the transformation law\eqref{susytransf}.  Then
  \begin{align}
    \delta \int d^2\theta \Phi_1 \Phi_2 =\int d^2 \theta [(\delta \Phi_1) \Phi_2 + \Phi_1(\delta \Phi_2)]\label{quadratictotalderivative}
  \end{align}
  is a total derivative.  Here integral $\int d^2\theta$ is defined by
  \begin{align}
    \int d^2\theta :=\int d\theta^{+} \int d\theta^{-}.
  \end{align}
\end{theorem}
\begin{proof}
  Because of an identity
  \begin{align}
    (\del_x \del_y \del_z \Phi_1) \Phi_2 + \Phi_1 (\del_x \del_y \del_z \Phi_2) &= \del_{x}((\del_y\del_z \Phi_1) \Phi_2)-\del_y((\del_z\Phi_1)(\del_{x}\Phi_1))+\del_z(\Phi_1(\del_x\del_y\Phi_2))\nonumber\\
    &=(\text{total derivative}),\label{integrationbyparts}
  \end{align}
  the differential operator $\susyder$ in \eqref{susytransf} satisfies the relation
  \begin{align}
    (\susyder \Phi_1) \Phi_2 +\Phi_1(\susyder \Phi_2)=
    (\text{total derivative in superspace}).
  \end{align}
  Therefore the left-hand side of \eqref{quadratictotalderivative} becomes
  \begin{equation}
    \begin{aligned}
      \delta \int d^2 \theta \Phi_1 \Phi_2 
      &=\int d^2 \theta [(\delta \Phi_1) \Phi_2 + \Phi_1(\delta \Phi_2)]\\
      &=\int d^2 \theta [(\susyder \Phi_1) \Phi_2 + \Phi_1(\susyder \Phi_2)]\\  
      &=\int d^2 \theta (\text{total derivative in superspace})\\
      &=(\text{total derivative in spacetime}).
    \end{aligned}    
  \end{equation}    
\end{proof}

Let us write down the action of the supersymmetric $\varphi$ theory with the help of theorem \ref{thm}.
The field in this theory is a real bosonic superfield $\Phi$ with periodicity  $\Phi\sim \Phi+2\pi$.  $\Phi$ is
expressed by components as
\begin{align}
  \Phi
  =\phi
  +i\theta^{+}\psi_{+}
  +i\theta^{-}\psi_{-}
  +i\theta^{+}\theta^{-}f.
\end{align}
The periodicity of $\Phi$ implies the periodicity of the component $\phi$ as $\phi\sim \phi + 2\pi$.  The Lagrangian density is written as
\begin{align}
  \Lcal = \int d^2\theta
    \frac{\mu_0}{2} \Dd_{-}\Phi \Dd_{+} \Phi,
\end{align}
where $\mu_0$ is a real parameter.  The supersymmetry transformation of this Lagrangian density becomes a total derivative according to theorem \ref{thm}. Therefore this theory has the supersymmetry.  The Lagrangian density is written by components as
\begin{align}
  \Lcal= \frac{\mu_0}{2}[
    4\del_{+}\phi \del_{-}\phi+2i\psi_{+}\del_{-}\psi_{+}+2i\psi_{-}\del_{+}\psi_{-} + f^2].\label{susyphitheory}
\end{align}
One can integrate out the auxiliary field $f$.  Then, the theory is the sum of the $\varphi$ theory \eqref{phitheory} and the $\psi$ theory \eqref{psitheory}.

In this paper, we only consider free field theories.  It seems difficult to construct an interacting theory with the supersymmetry considered in this paper.  
One can imagine that it is not easy to construct a theory with symmetry \eqref{Pxyz}, which is necessary to have the supersymmetry in this paper.  
It is an interesting future problem to construct such an interacting theory.

Here let us make some comments on the supersymmetrization of the $\phi$ theory in $2+1$ dimensions \cite{Seiberg:2020bhn}.  This theory is also similar to the free scalar theory in $1+1$ dimensions.  However, the relevant formula for the integration by parts is
\begin{align}
  (\del_{x}\del_{y}\Phi_1)\Phi_2-\Phi_1(\del_{x}\del_{y}\Phi_2)=
  (\text{total derivative}),
\end{align}
instead of \eqref{integrationbyparts}.  Therefore the procedure cannot be parallel to $1+1$ dimensions, and one have to make some trick to supersymmetrize the $\phi$ theory in $2+1$ dimensions.  This construction of the supersymmetric $\phi$ theory in $2+1$ dimensions is also an interesting future issue.

\subsection{Quadrupole symmetry current and the ground state degeneracy}
In our supersymmetric $\varphi$ theory, we have momentum and winding quadrupole symmetry currents as superfields.  Let us consider superfields
\begin{align}
  \Jq_{\pm}=\Dd_{\pm}\Phi=\psi_{\pm}+\theta^{\pm}J_{\pm},\qquad
  J_{\pm}=-2\del_{\pm}\phi,
\end{align}
where we use equations of motion in the expression by components.  Since the equations of motion can be expressed as $\Dd_{+}\Dd_{-}\Phi=0$, these currents satisfy the conservation law
\begin{align}
  \Dd_{\mp}\Jq_{\pm}=0 \quad
  \Leftrightarrow \quad \del_{\mp}\psi_{\pm}=0,\ \del_{\mp}J_{\pm}=0.
\end{align}
By the same argument as the $\varphi$-theory \cite{Gorantla:2020xap}, we find that
\begin{align}
  \oint dx \psi_{\pm},\qquad
  \oint dy \psi_{\pm},\qquad
  \oint dz \psi_{\pm},\qquad
  \oint dx J_{\pm},\qquad
  \oint dy J_{\pm},\qquad
  \oint dz J_{\pm},
\end{align}
are all conserved.  Actually, the bosonic charges of this theory are nothing but the momentum and winding quadrupole charges of the $\varphi$ theory \cite{Gorantla:2020xap}. The momentum and winding quadrupole charge density in \cite{Gorantla:2020xap} are expressed, respectively, as
\begin{align}
  \mu_{0}\del_{t}\phi = -\frac{\mu_0}{2}(J_{+}+J_{-}),\qquad
  \frac{1}{2\pi}\del_{x}\del_{y}\del_{z}\phi=\frac{1}{4\pi\alpha}(J_{-}-J_{+}).
\end{align}

Besides these bosonic charges, we also have the same number of fermionic quadrupole charges in our theory.  Let $A$ be the number of bosonic charges.  By the argument of \cite{Gorantla:2020xap}, $A$ is proportional to the area if we regularize the theory by the lattice.
Let us assume that we have a regularization that preserves the supersymmetry.  Then we have the same number $A$ of fermionic charges.  Let $\gamma_i,\ (i=1,\dots,A)$ be these fermionic charges.  We can choose the basis so that the canonical anti-commutation relations become
\begin{align}
  \{\gamma_i,\gamma_j\}=2\delta_{ij}.
\end{align}
This is nothing but the Clifford algebra.  Since these $\gamma_i$'s commute with the Hamiltonian, the space of ground states must be a representation space of this Clifford algebra. Therefore the ground state degeneracy is $2^{[A/2]}$.  We conclude that the residual entropy $[A/2]\log 2$ shows a sub-extensive behavior; it is proportional to the area instead of the volume.

In appendix \ref{app:lattice}, we consider a fermionic lattice model that may be a regularization of the $\psi$ theory.  We find that the number of fermionic charges and the residual entropy are proportional to the area.

\subsection{Duality}

Here we will derive the supersymmetric version of self-duality given in \cite{Gorantla:2020xap}. As a warm up, we rederive the self-duality of the $\varphi$ theory in our notation.  We start from the Lagrangian density
\begin{align}
  \Lcal = 2\mu_0 b_{-}b_{+}-2\bt_{-}(\del_{+}\phi-b_{+})+2(\del_{-}\phi-b_{-})\bt_{+}, \label{baseLagrangian}
\end{align}
where $\phi, b_{\pm},\bt_{\pm}$ are real bosonic fields. $\phi$ has periodicity $\phi\sim \phi + 2\pi$.

If we integrate out $\bt_{\pm}$ in \eqref{baseLagrangian}, we obtain constraints $\bt_{\pm}=\del_{\pm}\phi$. Then we integrate out $b_{\pm}$ and obtain the $\varphi$ theory \eqref{phitheory}.

On the other hand, if we first integrate out $b_{\pm}$, we obtain
\begin{align}
  \Lcal=\frac{2}{\mu_{0}}\bt_{-}\bt_{+}-2 \bt_{-}\del_{+}\phi + 2 \del_{-}\phi \bt_{+}.
\end{align}
Then we integrate out $\phi$ and obtain the constraint
\begin{align}
  \del_{-}\bt_{+}-\del_{+}\bt_{-}=0, \label{btconstraint}
\end{align}
as a necessary condition.
This constraint can be solved by introducing a real bosonic field $\phit$ as
\begin{align}
  \bt_{\pm}=c\del_{\pm}\phit, \label{tempsolbt}
\end{align}
where $c$ is a constant that is determined so that $\phit$ is $2\pi$ periodic.

Let us determine $c$. Besides the constraint \eqref{btconstraint}, we obtain some additional constraints from the $2\pi$ periodicity of $\phi$. Consider the Fourier modes of the fields
\begin{align}
  \phi(t,\xv)=\sum_{\kv}\phi_{\kv}(t)e^{i\kv\cdot \xv},\qquad
  \bt_{\pm}(t,\xv)=\sum_{\kv}\bt_{\pm,\kv}(t)e^{i\kv\cdot \xv}.
\end{align}
The $2\pi$ periodicity of $\phi$ implies only the periodicity of the zero mode $\phi_{\zv}\sim \phi_{\zv}+2\pi$.  Let us focus on this zero mode.  The part of the action including this $\phi_{\zv}$ is 
\begin{align}
  S=\int dt V(\bt_{+,\zv}-\bt_{-,\zv})\del_{t}\phi_{\zv}+\cdots,
\end{align}
where $V$ is the volume of the space $V=\int d^3\xv 1$.
The canonically conjugate momentum of $\phi_{\zv}$ is given by
\begin{align}
  p_{\zv}=V(\bt_{+,\zv}-\bt_{-,\zv}). \label{conjugate-momentum}
\end{align}
Because of the periodicity of $\phi_{\zv}$, the eigenvalues of $p_{\zv}$ must be integers. From the relation \eqref{tempsolbt}, we obtain
\begin{align}
  \bt_{+}-\bt_{-}=c\alpha \del_{x}\del_{y}\del_{z}\phit.
\end{align}
Let us integrate both sides of this equation in the space $\int d^3\xv$. The right-hand side becomes the conjugate momentum \eqref{conjugate-momentum}.  The left-hand side is expressed in terms of the total winding number $\wt$
\begin{align}
  \wt=\frac{1}{2\pi}\int d^3\xv \del_{x}\del_{y}\del_{z}\phit.
\end{align}
This $\wt$ must be an integer due to the $2\pi$ periodicity of $\phit$.
Therefore we obtain the relation
\begin{align}
  p_{\zv}=2\pi c \alpha \wt.
\end{align}
Since both $p_{\zv}$ and $\wt$ can take all integer values, we can determine the constant $c$ as
\begin{align}
  c=\frac{1}{2\pi \alpha}.
\end{align}
The expression \eqref{tempsolbt} become
\begin{align}
  \bt_{\pm}=\frac{1}{2\pi \alpha} \del_{\pm}\phit.
\end{align}
Finally we obtain the dual Lagrangian
\begin{align}
  \Lcal = 2\mut_{0}\del_{-}\phit \del_{+} \phit, \qquad \mut_{0}=\frac{1}{(2\pi \alpha)^2\mu_0}.
\end{align}
This is just the same result as \cite{Gorantla:2020xap}. Notice that $\mu$ in \cite{Gorantla:2020xap} is expressed as $\mu=\frac{1}{\mu_0 \alpha^2}$.

Next let us turn to the duality of supersymmetric $\varphi$ theory.  Starting from the action
\begin{align}
  S= \int dt \int d^3\xv  \int d^2\theta \Ls,\qquad
  \Ls= \frac{\mu_0}{2}B_{-}B_{+} + \frac12 \Bt_{-} (\Dd_{+}\Phi-B_{+})-\frac12 (\Dd_{-}\Phi-B_{-})\Bt_{+},\label{baseSUSYLagrangian}
\end{align}
where $B_{\pm},\Bt_{\pm}$ are real fermionic superfields and $\Phi$ is a real bosonic superfield with the periodicity $\Phi\sim \Phi+2\pi$.

If we first integrate out $\Bt_{\pm}$, we obtain the constraint $B_{\pm}=\Dd_{\pm}\Phi$. Then we integrate out $B_{\pm}$ and obtain the supersymmetric $\varphi$-theory \eqref{susyphitheory}.

On the other hand, if we first integrate out $B_{\pm}$, we obtain
\begin{align}
  \Ls= \frac{1}{2\mu_0}\Bt_{-}\Bt_{+} + \frac12 \Bt_{-}\Dd_{+}\Phi-\frac12 \Dd_{-}\Phi\Bt_{+}.
\end{align}
Then we integrate out $\Phi$ and obtain the constraint
\begin{align}
  \Dd_{-}\Bt_{+}+\Dd_{+}\Bt_{-}=0,
\end{align}
which is solved in terms of a bosonic superfield $\Phit$ as
\begin{align}
  \Bt_{\pm}=\frac{1}{2\pi\alpha}\Dd_{\pm}\Phit.
\end{align}
The coefficient is determined so that $\Phit$ has the periodicity $\Phit\sim \Phit+2\pi$.  Finally we obtain the dual theory as
\begin{align}
  \Ls = \frac{\mut_0}{2}\Dd_{-}\Phit \Dd_{+}\Phit,\qquad \mut_{0}=\frac{1}{(2\pi \alpha)^2\mu_0}.
\end{align}

Notice that we only use the quadratic actions during this procedure, and therefore the supersymmetry is manifest.

\section{Supersymmetric tensor gauge theory}
\label{sec:gaugetheory}

In this section, we supersymmetrize the tensor gauge theory.  We write down the supersymmetric tensor gauge theory action.  We also discuss BPS Wilson line defects and fractons.

\subsection{Gauge superfields}
Let us consider gauging the shift symmetry of $\Phi$ in the supersymmetric $\varphi$ theory\eqref{susyphitheory}.  The parameter of the shift symmetry is promoted to a real bosonic superfield $K$ with periodicity $K\sim K+2\pi$.  The gauge transformation of $\Phi$ is given by
\begin{align}
  \Phi \to \Phi'=\Phi + K.
\end{align}
Let us introduce real fermionic superfields $\Gamma_{\pm}$ and the super covariant derivatives $\nabla_{\pm}$ by
\begin{align}
  \nabla_{\pm}\Phi:=\Dd_{\pm}\Phi - \Gamma_{\pm}.
\end{align}
The gauge transformation law of $\Gamma_{\pm}$ are determined so that $\nabla_{\pm}\Phi$ are gauge invariant, and given by  
\begin{align}
  \Gamma_{\pm}\to \Gamma'_{\pm}:=\Gamma_{\pm}+\Dd_{\pm}K. \label{gaugetransf}
\end{align} 

Let us look at the components of $\Gamma_{\pm}$ and their gauge transformation.  First, let us denote the components of $K$ as
\begin{align}
  K=\omega + i\theta^{+}\eta_{+}+i\theta^{-}\eta_{-}+i\theta^{+}\theta^{-}\tau,
\end{align}
where $\omega, \tau$ are real bosonic fields and $\eta_{\pm}$ are real fermionic fields.  We also denote the components of $\Gamma_{\pm}$ as
\begin{equation}
  \begin{aligned}
    &\Gamma_{+}=\chi_{+}-2\theta^{+}A_{+}+\theta^{-}(B+\sigma)-2i\theta^{+}\theta^{-}(\lambda_{+}+\del_{+}\chi_{-}),\\
    &\Gamma_{-}=\chi_{-}-\theta^{+}(B-\sigma)-2\theta^{-}A_{-}+2i\theta^{+}\theta^{-}(\lambda_{-}+\del_{-}\chi_{+}),
  \end{aligned}
  \label{gammacomponents}  
\end{equation}
where $B,\sigma, A_{\pm}$ are real bosonic fields, and $\chi_{\pm},\lambda_{\pm}$ are real fermionic fields.  Notice that these are the most generic form of two real fermionic superfields.  The gauge transformations \eqref{gaugetransf} of the components read
\begin{align}
  \chi'_{\pm}=\chi_{\pm}+\eta_{\pm},\qquad B'=B+\tau,\qquad A'_{\pm}=A_{\pm}+\del_{\pm}\omega,\qquad \sigma'=\sigma,\qquad \lambda'_{\pm}=\lambda_{\pm}.
\end{align}

We can construct the gauge invariant superfield $\Sigma$ from $\Gamma_{\pm}$ as
\begin{align}
  \Sigma = \frac{i}{2}(\Dd_{+}\Gamma_{-}+\Dd_{-}\Gamma_{+}).
\end{align}
In order to find the component expression of $\Sigma$, it is convenient to choose Wess-Zumino (WZ) gauge 
\begin{align}
  \chi_{\pm}=0,\ B=0. \label{WZgauge}
\end{align}
In this gauge, $\Gamma_{\pm}$ are expressed as
\begin{align}
  &\Gamma_{+}=-2\theta^{+}A_{+}+\theta^{-}\sigma-2i\theta^{+}\theta^{-}\lambda_{+},\qquad
  \Gamma_{-}=\theta^{+}\sigma-2\theta^{-}A_{-}+2i\theta^{+}\theta^{-}\lambda_{-}.  
\end{align}
We obtain the expression of $\Sigma$ in terms of components as
\begin{align}
  \Sigma = \sigma + i\theta^{+}\lambda_{+}+i\theta^{-}\lambda_{-}+i\theta^{+}\theta^{-}2F_{+-},\qquad
  F_{+-}:=\del_{+}A_{-}-\del_{-}A_{+}.
\end{align}

Here let us mention the relation to the tensor gauge field in \cite{You:2019cvs,You:2019bvu,Gorantla:2020xap}.  The gauge fields $A_{\pm}$ and the field strength $F_{+-}$ in this paper are related to the gauge fields $B_0,B_{xyz}$ and the field strength $E_{xyz}$ in \cite{Gorantla:2020xap} as
\begin{align}
  A_{\pm}=\frac12 ( B_{0}\pm \alpha B_{xyz}), \qquad
  F_{+-}=-\frac{\alpha}{2}(\del_{t}B_{xyz}-\del_{x}\del_{y}\del_{z}B_{0})
  =-\frac{\alpha}{2}E_{xyz}.
\end{align}
Thus the superfield $\Gamma_{\pm}$ are supersymmetry completions of the tensor gauge fields $B_0, B_{xyz}$.

\subsection{Gauge theory action}

Here we construct the supersymmetric gauge theory action.
The gauge field part is written by the gauge invariant superfield $\Sigma$ and consists of the kinetic term and the potential term.  The potential term must be quadratic polynomial or linear function in order to preserve the supersymmetry in our formalism.

The kinetic term in the Lagrangian density is written as
\begin{align}
  \Lcal_{\text{kinetic}}=\int d^2\theta \frac{1}{g_e^2\alpha^2}\Dd_{-}\Sigma \Dd_{+}\Sigma=\frac{1}{g_e^2\alpha^2}
  (4\del_{+}\sigma\del_{-}\sigma+2i\lambda_{+}\del_{-}\lambda_{+}+2i\lambda_{-}\del_{+}\lambda_{-}+(2F_{+-})^2),
\end{align}
where $g_e$ is a real coupling constant.
We find that this theory includes a boson $\sigma$ with the same kinetic term as the $\varphi$ theory\eqref{phitheory}, and $\lambda_{\pm}$ that are a copy of the $\psi$ theory \eqref{psitheory} in addition to the gauge theory kinetic term.

The quadratic term is written as
\begin{align}
  \Lcal_{2}=\int d^2\theta im\Sigma^2
  =4m\sigma F_{+-}-2im\lambda_{+}\lambda_{-},
\end{align}
where $m$ is a real constant.  This term includes axion-like interaction between $\sigma$ and $A_{\pm}$, as well as a fermion mass term.

The linear term is the theta term
\begin{align}
  \Lcal_{1}=\int d^2 \theta \left(-\frac{i\vartheta}{2\pi\alpha}\Sigma\right)=-\frac{\vartheta}{2\pi \alpha}2F_{+-},
\end{align}
where $\vartheta$ is a real constant with the periodicity $\vartheta\sim \vartheta+2\pi$.

We can add the supersymmetric $\varphi$ theory as a matter to this gauge theory.
The matter part of the action is obtained by replacing the derivatives $\Dd_{\pm}$ with covariant the derivatives $\nabla_{\pm}$ as
\begin{align}
  \Lcal_{\text{matter}} &= \int d^2\theta \frac{\mu_0}{2}\nabla_{-}\Phi \nabla_{+}\Phi.
\end{align}
One can expand this action by components and obtain the fully gauge invariant action by components.

Instead of writing down the gauge unfixed action in components, we choose WZ gauge \eqref{WZgauge} and write down the action in components.
In WZ gauge the matter part of the action is given by
\begin{align}
  \Lcal_{\text{matter}} 
  =\frac{\mu_0}{2}[4D_{+}\phi D_{-}\phi + 2i \psi_{+}(\del_{-}\psi_{+}-\lambda_{-})
  +2i\psi_{-}(\del_{+}\psi_{-}-\lambda_{+})+f^2-\sigma^2],
  \quad D_{\pm}\phi:=\del_{\pm}\phi -A_{\pm}. \label{matteractionWZ}
\end{align}
In this action, $\lambda_{\pm}$ couple to the fermionic quadrupole currents $\psi_{\mp}$, and therefore one can interpret that $\lambda_{\pm}$ are fermionic analogs of the tensor gauge fields.

Finally, let us make some comments on the supersymmetry of the action \eqref{matteractionWZ} in WZ gauge.
Since the supersymmetry transformation \eqref{susytransf} breaks WZ gauge condition \eqref{WZgauge}, the matter action \eqref{matteractionWZ} in WZ gauge is not invariant under the transformation \eqref{susytransf}.
Therefore we should consider the supersymmetry transformation \eqref{susytransf} combined with the gauge transformation $\tilde{\delta}$ with the parameter
\begin{align}
  K
  =
  i\theta^{+}(\epsilon_{+}\sigma+2\epsilon_{-}A_{+})
  +i\theta^{-}(-\epsilon_{-}\sigma-2\epsilon_{+}A_{-})
  +i\theta^{+}\theta^{-}(-i\epsilon_{-}\lambda_{+}-i\epsilon_{+}\lambda_{-}).
\end{align}
We denote this combined transformation by $\delta'=\delta + \tilde{\delta}$.  This improved supersymmetry transformation $\delta'$ keeps WZ condition and makes the action \eqref{matteractionWZ} invariant.  This $\delta'$ transformation of components is given by
\begin{equation}
  \begin{aligned}
    &\delta' A_{+}=-i\epsilon_{+}\lambda_{+},\qquad
    &&\delta' A_{-}=i\epsilon_{-}\lambda_{-},\\
    &\delta' \lambda_{+}=-2\epsilon_{-}\del_{+}\sigma-2\epsilon_{+}F_{+-},\qquad
    &&\delta' \lambda_{-}=2\epsilon_{+}\del_{-}\sigma-2\epsilon_{-}F_{+-},\\
    &\delta' \sigma =i\epsilon_{-}\lambda_{+}-i\epsilon_{+}\lambda_{-}.&&\\
    &\delta'\phi = i\epsilon_{-}\psi_{+}-i\epsilon_{+}\psi_{-},&&\\
    &\delta'\psi_{+} = -2\epsilon_{-}D_{+}\phi-\epsilon_{+}(f-\sigma),\qquad
    &&\delta'\psi_{-} = 2\epsilon_{+}D_{-}\phi-\epsilon_{-}(f+\sigma),\\
    &\delta' f=
    i\epsilon_{-}(2\del_{+}\psi_{-}- \lambda_{+})
    +i\epsilon_{+}(2\del_{-}\psi_{+}- \lambda_{-}),
    &&D_{\pm}\phi:=\del_{\pm}\phi -A_{\pm}.
  \end{aligned}
\end{equation}
One can explicitly check that $\delta'$ transformation of $\Lcal_{\text{matter}}$ is a total derivative.

\subsection{BPS defects as fractons}
We have a BPS defects in the supersymmetric tensor gauge theory constructed above.  Let us consider the Wilson line
\begin{align}
  \exp\left(i\int dt (A_{+}+A_{-}+\sigma)\right).
\end{align}
This Wilson line is one in \cite{Gorantla:2020xap} dressed by $\sigma$.
This Wilson line is invariant under the gauge transformation \eqref{gaugetransf}.  It is also invariant under a half of the supersymmetry transformation with parameters $\epsilon_{-}=\epsilon_{+}$, and therefore 
we call it as a BPS Wilson line.  This Wilson line cannot move because of the gauge symmetry.  Therefore this Wilson line describes the probe limit of a fracton.

In \cite{Gorantla:2020xap}, it has been shown that once four of these defects form a quadrupole, it can move collectively.  
This is also true in our defect in our supersymmetric theory.  However, all the supersymmetry is broken if the quadrupole moves even if the velocity is constant.  
It is an interesting future problem to find a nice dressing to make the moving quadrupole BPS, or some no-go theorem.

\subsection*{Acknowledgement}
This work was supported in part by JSPS KAKENHI Grant Number 15K05054.

\appendix

\section{Lattice model}
\label{app:lattice}
Here we would like to consider a naive lattice regularization of the $\psi$ theory.
We consider 3-dimensional cubic lattice. Each site is labeled by $\nv=(n_x,n_y,n_z) \in \Zb_{L_x}\times\Zb_{L_y}\times\Zb_{L_z}$.  We impose periodic boundary conditions $n_i\sim n_i+L_i,\ (i=x,y,z)$ and assume $L_i$ are even integers for simplicity.  We assign a hermitian fermionic operator $c_{\nv}$ to each site.  They satisfy anti-commutation relation
\begin{align}
\{c_{\nv},c_{\nv'}\}=\delta_{\nv,\nv'}.
\end{align}

Let us consider the Hamiltonian with a real constant parameter $r$ written as
\begin{equation}
  \begin{aligned}
    H=&r\sum_{\nv}c_{\nv} i \Delta_{xyz}c_{\nv},\\
    \Delta_{xyz}c_{\nv}:=&\frac{1}{8}\sum_{\vec{m}\in \{+1,-1\}^3}
    m_x m_y m_z c_{\nv+\vec{m}}\\
    =&\frac18( 
      c_{(n_x+1,n_y+1,n_z+1)}
      -c_{(n_x-1,n_y+1,n_z+1)}
      -c_{(n_x+1,n_y-1,n_z+1)}
      +c_{(n_x-1,n_y-1,n_z+1)}\\
    & -c_{(n_x+1,n_y+1,n_z-1)}
      +c_{(n_x-1,n_y+1,n_z-1)}
      +c_{(n_x+1,n_y-1,n_z-1)}
      -c_{(n_x-1,n_y-1,n_z-1)}
    ).\label{latticemodel}
  \end{aligned}    
\end{equation}
One may naively expect that one can take the continuum limit so that $\Delta_{xyz}\to a^3 \del_x\del_y\del_z$ and the system is well-described by (half of ) the $\psi$ theory.  If we look more closely at this system, we find so-called ``doublers'' and this system seems to contain 4 copies of the $\psi$ theory.  This picture may be still too naive since short distance physics in $x$-direction affects low energy physics if the wavelength in $y$- or $z$-direction is very large.

This system \eqref{latticemodel} is solved by the Fourier transformation
\begin{equation}
  \begin{aligned}
    &c_{\nv}=\frac{1}{\sqrt{N}}\sum_{\kv}b_{\kv}e^{i\kv\cdot \nv},\quad N:=L_xL_yL_z,\\
    &\kv=(k_x,k_y,k_z),\quad k_i\in \frac{2\pi}{L_i}\Zb,\quad k_i\sim k_i+2\pi,\quad (i=x,y,z).
  \end{aligned}    
\end{equation}
$b_{\kv}$ satisfy anti-commutation relations
\begin{align}
  \{b_{\kv},b_{\kv'}\}=\delta_{\kv,-\kv'}.
\end{align}
The Hamiltonian \eqref{latticemodel} becomes
\begin{align}
  H=r\sum_{\kv} \sin k_x \sin k_y \sin k_z b_{-\kv}b_{\kv}.
\end{align}

Let us count the ground state degeneracy of the system \eqref{latticemodel}.  
The number of fermionic zero modes is the number of $\kv$ satisfying $\sin k_x \sin k_y \sin k_z=0$.  Thus it is 
\begin{align}
  2L_xL_y+2L_yL_z+2L_zL_x-4L_x-4L_y-4L_z+8,
\end{align}
and the ground state degeneracy is given by
\begin{align}
  2^{L_xL_y+L_yL_z+L_zL_x-2L_x-2L_y-2L_z+4}.
\end{align}
We conclude that the residual entropy is proportional to the area instead of the volume of the system.

\bibliographystyle{utphys}
\bibliography{ref}
\end{document}